\documentclass[submission,copyright,creativecommons]{eptcs}

\usepackage[latin1]{inputenc}

\usepackage{booktabs}
\usepackage{amssymb}
\usepackage{amsfonts}
\usepackage{amssymb}
\usepackage{xspace}
\usepackage{graphicx}
\usepackage{amsbsy}
\usepackage{url}
\usepackage{program}
\usepackage{nccbbb}
\usepackage{comment}
\usepackage[ruled]{algorithm}

\newcommand{\uppaal}{\textsc{Uppaal}\xspace}

\newcommand{\forget}[1]{}

\newcommand{\rplus}{\bbbr_{\geq0}}

\newtheorem	{definition}{Definition}
\newtheorem{theorem}{Theorem}

\title{Computing Nash Equilibrium in Wireless Ad~Hoc Networks: A Simulation-Based Approach
\thanks {The paper is supported by the IDEA4CPS project, the Artist MBAT project and VKR Center of Excellence MT-LAB.}
}

\author{Peter~Bulychev
\institute{Aalborg University, Denmark}
\email{pbulychev@cs.aau.dk}
\and
Alexandre~David 
\institute{Aalborg University, Denmark}
\email{adavid@cs.aau.dk}
\and
Kim~G.~Larsen
\institute{Aalborg University, Denmark}
\email{kgl@cs.aau.dk}
\and
Axel~Legay
\institute{INRIA Rennes/Aalborg University}
\email{axel.legay@irisa.fr}
\and
Marius~Miku\v{c}ionis
\institute{Aalborg University, Denmark}
\email{marius@cs.aau.dk}
}
\def\titlerunning{Computing Nash Equilibrium in Wireless Ad~Hoc Networks: A Simulation-Based Approach}
\def\authorrunning{P. Bulychev, A. David, K.G. Larsen, A. Legay, M. Miku\v{c}ionis}
\begin{document}
\maketitle





  \maketitle

  \begin{abstract}  
    This paper studies the problem of computing Nash equilibrium in
    wireless networks modeled by Weighted Timed Automata.  Such
    formalism comes together with a logic that can be used to describe
    complex features such as timed energy constraints. Our
    contribution is a method for solving this problem using
    Statistical Model Checking.  The method has been implemented in
    \uppaal model checker and has been applied to the analysis of Aloha
    CSMA/CD and IEEE 802.15.4 CSMA/CA protocols.

  \end{abstract}

\section{Introduction}

One of the important aspects in designing wireless ad-hoc networks is
to make sure that a network is robust to the selfish behavior of its
participants.  This problem can be formulated in terms of a game
considering that network nodes behave in a rational way and want to
maximize their utility.  A wireless network is robust iff its
configuration satisfies Nash equilibrium~(NE), i.e. it is not
profitable for a node to alter its behavior to the detriment of other
nodes.

In this paper we propose a new methodology to compute NE in wireless
ad-hoc networks. Our approach is based on Statistical Model
Checking~(SMC)~\cite{SMC_Yonas, SMC_Legay}, an approach used in the
formal verification area. SMC has a wide range of applications in the
areas such as systems biology or automotive.  The core idea of SMC is
to monitor a number of simulations of a system and then use the
results of statistics (e.g. sequential analysis) to get an overall
estimate of the probability that the system will behave in some
manner. Thus SMC can help to overcome the undecidability issues, that arise in the formal analysis of wireless ad-hoc networks \cite{ADRST-formats11}.
While the idea ressembles the one of classical Monte Carlo
simulation, it is based on a formal semantics of systems that allows us
to reason on very complex behavioral properties of systems (hence the
terminology). This includes classical reachability property such as
``can i reach such a state?'', but also non trivial properties such as
``can i reach this state x times in less than y units of time?''.

Here we use a semantics for systems that is based on timed automata. We
assume that elements of a network are modeled using Weighted Timed
Automata (WTA), that is a model for timed system together with a
stochastic semantics. The model permits, for example, to describe
stochastically how the behaviors of a system involve with respect to
time.  As an example, probability can be used to say that the system
is more likely to move to the next state in five units of time rather
than in ten. Our approach permits to describe arbitrary distributions
when combining individual components. In addition, WTA are equiped
with nice communication primitives between components, including
e.g. message passing. The fact that we rely on WTA allows us to
describe utility functions with cost constraint temporal logic called
PWCTL. PWCTL is a logic that allows to temporally quantify on
behaviors of components as well as on their individual features
(cost, energy consumption). By the existing results, we know that one
can always define a probability measure on sets of run of a WTA that
satisfy a given property written in cost constraint temporal
logic. The latter clearly arises the importance of defining a formal
semantics for system's models.

Going back to modeling wireless ad-hoc network, we will assume that
each node can work according to one out of finitely many
configurations that we call strategies (a strategy being a choice of a
configuraiton).  We also assume that each network node has a goal and
a node's utility function is equal to the probability that this goal
will be reached on a random system run. Each node will be represented
with a WTA and a goal will be described with PWCTL. Our
algorithm for computing NE consists of two phases.  First, we apply
simulation-based algorithm to compute a strategy that most likely
(heuristic) satisfies NE. This is done by monitoring several runs of
the systems with respect to a property and then use classical Monte
Carlo ratio to estimate the probability that this strategy is indeed
the good one. In the second phase we apply statistics to test the
hypothesis that this strategy actually satisfies NE. Indeed, it is
well-known that such a NE may not exist~\cite{game_theory}, so our
statistical algorithm goes for the best estimate out of it that
corresponds to a relaxed NE for the system.

We implemented a distributed version of this algorithm that uses
\uppaal statistical model checker as a simulation
engine\,\cite{UPPAAL_SMC}. The \uppaal toolset offers a nice user
interface that makes it one of the most widely used formal
verification based tool in academia.  Thanks to the independent
simulations that the algorithm generates, this problem can easily be
parallelized and distributed on, e.g., PC clusters.

Finally, we apply our tool to the analysis of two probabilistic
CSMA~(Carrier Sense Multiple Access) protocols: k-persistent Aloha
CSMA/CD protocol and IEEE 802.15.4 CSMA/CA protocol. The two case
studies we present in this paper serve mainly demonstration purposes.
However, we are the first who study Nash Equilibrium in
\emph{unslotted} Aloha (the slotted version of Aloha was previously
studied in~\cite{Mackenzie2001} and~\cite{wirelessgametheory}).  Our
result that there exists only ``always transmit'' Nash Equilibrium
strategy in IEEE 802.15.4 CSMA/CA reproduces the analogous result for
the IEEE 802.11 CSMA/CA protocol~\cite{Cagalj04oncheating}.

The problem of computing Nash Equilibrium in wireless ad-hoc networks
was first considered in~\cite{Mackenzie2001} and surveid
in~\cite{wirelessgametheory}.  Most of existing works are based on an
analytical solution which does not scale well with complex
models~\cite{conf/pimrc/GhasemiF08,Mackenzie2001}. Similarly to us,
some other works \cite{Cagalj04oncheating} are based on a
simulation-based approach. However, contrary to us, they do not
assign any statistical confidence to the results, and they do not
take advantage of temporal logic to express arbitrary objectives. 

   
  
\section{Weighted Timed Automata}

In this section, we briefly recap the concept of Weighted Timed Automata
(WTA), see \cite{uppaal_formats} for more details. We denote
${\cal{B}}(X)$ to be a finite conjunction of bounds of the form $x\sim
n$ where $x\in X$, $n\in\bbbn$, and $\sim\in\{<,\leq, >, \geq \}$.

\begin{definition} A \emph{Weighted Timed Automaton}\footnote{In the
    classical      notion      of      priced      timed      automata
    \cite{HSCC01,Alur01optimalpaths}   cost-variables   (e.g.   clocks
    where  the rate  may differ  from $1$)  may not  be  referenced in
    guards,  invariants  or  in   resets,  thus  making  e.g.  optimal
    reachability decidable.  This is in contrast to our notion of WTA,
    which is as expressive  as linear hybrid systems \cite{CONCUR00}.} (WTA) is a
  tuple ${\cal A}=(L,\ell_0,X,E,R,I)$ where: (i) $L$ is a finite set
  of locations, (ii) $\ell_0\in L$ is the initial location, (iii) $X$
  is a finite set of real-valued variables called clocks, (iv)
  $E\subseteq L \times {\cal{B}}(X) \times 2^X \times L$ is a finite
  set of edges, $(v)$ $R:L\rightarrow\mathbb{Z}_{\geq 0}$
  assigns a rate vector to each location, and (vi)
  $I:L\rightarrow{\cal{B}}(X)$ assigns an invariant to each location.
\end{definition}

A state of a WTA is a pair $(l,v)$ that consists of a location $l$ and
a valuation of clocks $\nu : X \rightarrow \rplus$.  From a state
$(l,v) \in L \times \rplus^X$ a WTA can either let time progress or do
a discrete transition and reach a new location.  During time delay
clocks are growing with the rates defined by $R(l)$, and the resulting
clock valuation should satisfy invariant $I(l)$.  A discrete
transition from $(l, v)$ to $(l', v')$ is possible if there is $(l, g,
Y, l') \in E$ such that $v$ satisfies $g$ and $v'$ is obtained from
$v$ by resetting clocks from the set $Y$ to $0$.  A run of WTA is a
sequence of alternating time and discrete transitions. Several WTA
$M_1, M_2, \dots, M_n$, can communicate via inputs and outputs to
generate Networks of WTAs (NWTAs) $M_1 \| M_2 \| \dots \| M_n$.

 In our early works \cite{uppaal_formats}, the stochastic semantics of
 WTA components associates probability distributions on both the
 delays one can spend in a given state as well as on the transition
 between states.  In \uppaal uniform distributions are applied for
 bounded delays and exponential distributions for the case where a
 component can remain indefinitely in a state.  In a network of WTAs
 the components repeatedly race against each other, i.e.  they
 independently and stochastically decide on their own how much to
 delay before outputting, with the ``winner'' being the component that
 chooses the minimum delay. As observed in \cite{uppaal_formats}, the
 stochastic semantics of each WTA is rather simple (but quite
 realistic), arbitrarily complex stochastic behavior can be obtained
 by their composition when mixing individual distributions through
 message passing. The beauty of our model is that these distributions
 are naturally and automatically defined by the network of WTAs.

Our implementation supports extension of WTA, coming from the language
of the \uppaal model checker~\cite{UPPAAL}.  Such models can contain
integer variables that can be present in transition guards, and they
can be updated only when a discrete transition is taken.
Additionally, we support other features of the \uppaal model checker's
input language such as data structures and user-defined functions.

A parametrized WTA $M(p)$ is a WTA in which some integer constant
(transition weight or constant in variable assignment/clock invariant)
is replaced by a parameter $p$.

For defining properties we use cost-constraint temporal logic PWCTL,
which contains formulas of the form $\Diamond_{c\leq C}\phi$.  Here
$c$ is an observer clock (that is never reset and grows to
infinity on any infinite run of a WTA), $C \in \rplus$ is a constant and $\phi$ is a
state-predicate.  We say that a run $\pi$ satisfies $\Diamond_{c\leq
  C}\phi$ if there exists a state $(l, v) \in \pi$ in this run such that it satisfies
$\phi$ and $v(c) \leq C$.  We define $Pr[M \models \psi]$ to be equal
to the probability that a random run of $M$ satisfies $\psi$.
    
\section{Modeling Formalism and problem statement}

We consider that each node operates according to one out of finitely
many configurations.  Thus a network of $N$ nodes can be modelled by:
\begin{equation}
S(p_1, p_2, \dots, p_N) \equiv M(p_1) \| M(p_2) \| \dots \| M(p_N) \|
C
\end{equation}
where $M$ is a parametrized model of a node, $p_i \in P$ (i.e, the
behavior of a node relies on some value - strategy - assigned to the
parameters), $P$ is a finite set of configurations and $C$ is a model
of a medium. 

Consider a parameterized NWTA $S(p_1, p_2, \dots, p_N)$ that models a
wireless network of $N$ nodes.  Here each $p_i$ defines a
configuration of a node $i$ and ranges over a finite domain $P$.
Since the players (nodes) are symmetric, we can analyze the game from the point of view of the first node only.
Thus we will consider the goal of the first node only, and this goal is defined by a PWCTL formula $\psi$.

We can view a system as a game $G=(N, P, U)$, where $N$ is a number of
players~(nodes), $P$ is a set of strategies~(parameters) and $U : P^N
\rightarrow [0,1]$ is an utility function of the first player defined
as
\begin{equation}
  U(p_1, p_2, \dots, p_N) \equiv Pr[S(p_1, p_2, \dots, p_N) \models \psi]
\end{equation}

We consider that there is a master
node that knows the network configuration (here the number of nodes) 
and broadcasts the strategy (parameter) that all the nodes should use.

If all the nodes are honest, they will play according to the strategy proposed by the master node.
Thus in this case the master node should use a symmetric optimal strategy, i.e. a strategy $p$ such that for all other strategies $p'$ we have $U(p, p, \dots, p) \geq U(p', p', \dots, p')$.

However, if there are selfish nodes, they might deviate from the symmetric optimal strategy to increase the value of their utility functions (and the rest of the nodes can possibly suffer from that).
Thus we will consider a Nash Equilibrium strategy that is stable with respect to the behavior of such selfish nodes~(but possibly this strategy is less efficient than the symmetric optimal one).
More formally, a strategy $p$ is said to be a Nash Equilibrium (NE), iff for all
$p' \in P$ we have $U(p, p, \dots, p) \geq U(p', p, \dots, p)$.  

However, a Nash Equilibrium may not exist~\cite{game_theory}\footnote{Note, that we assume only non-mixed (pure) strategies.}, thus in this paper we will
consider a \emph{relaxed} definition of Nash Equilibrium.

\begin{definition}\label{nedef}
 A strategy $p$ satisfies symmetric $\delta$-relaxed NE iff for all
 $p' \in P$ we have $U(p, p, \dots, p) \geq \delta \cdot U(p', p,
 \dots, p)$.
\end{definition}

The value of $\delta$ measures the quality of a strategy $p$.  If
$\delta \geq 1$, then $\delta$-relaxed NE satisfies the traditional
definition of the (non-relaxed) NE.  Otherwise, if $1-\delta$ is small,
then we can conclude that a node's gain of switching is negligible and
it'll stick with the $\delta$-relaxed NE strategy. 
A $\delta$-relaxed NE strategy can be also used when a set of possible strategies is infinite.
In this case we can discritize this set (approximate it by a \emph{finite} set of strategies) and search for a $\delta$-relaxed NE strategy 
in this finite set. If an utility function is smooth, then this strategy can be a good approximation for a NE in the original (infinite) set of strategies.

In this paper, we will solve the problem of searching for a strategy
that satisfies $\delta$-relaxed NE for as large $\delta$ as it is
possible.

For readability, in the rest of the paper we will write $U(p', p)$ and
$S(p', p)$ instead of $U(p', p, \dots, p)$ and $S(p', p, \dots, p)$,
respectively.
 

\section{Algorithm for Computing Nash Equilibrium}

One may suggest that in order to compute NE we can compute the values
of $U(p', p)$ for all pairs $(p', p)$ and then use definition
\ref{nedef} to compute the maximal value of $\delta$.  However, we
can't do that because the problem of evaluating PWCTL formula on a
model (i.e. computing $Pr[S \models \psi]$) is undecidable in
general for WTA~\cite{WMTL08}.

In this paper, we will use Statistical Model Checking (SMC)\,\cite{SMC_Legay}
based approach to overcome this undecidability problem.  The main idea
of this approach is to perform a large number of simulations and then
apply the results of statistics to estimate the probability that a
system satisfies a given property.

Our method of computing NE consists of two phases.  During the first
phase (presented in Section 4.1) we apply a simulation-based algorithm
to search for the best candidate $p$ for a Nash Equilibrium.  Then, in
the second phase (presented in Section 4.2), we apply statistics to
evaluate $p$, i.e. to find the maximum $\delta$ such that with a given
significance level we can accept the statistical hypothesis that $p$
is a $\delta$-relaxed Nash Equilibrium.

In the rest of the paper, we use straightforward simulation-based Monte Carlo method for computing estimations of utility function's values. 
In this method we perform $n$ random simulations of $S(p', p)$ for a given pair $(p', p)$ and count the number $k$ of how many simulations satisfied $\psi$.
Then we use the following estimation: $\widetilde{U}(p', p)=k/n$. 

\subsection{Finding a Candidate for a Nash Equilibrium}

\begin{algorithm}[t]
\caption{Computation Of The Best Candidate For Nash Equilibrium}
\label{alg:search}
\begin{small}
\underline{\textbf{Input}}: $P=\{p_i\}$ --- finite set of parameters, $U(p_i, p_k)$ --- utility function, $d \in [0,1]$ --- threshold \\
\underline{\textbf{Algorithm}}: \\
\verb|1. | \textbf{for every} $p_i \in P$ compute estimation $\widetilde{U}(p_i,p_i)$  \\
\verb|2. | $waiting$ := $P$ \\
\verb|3. | $candidates$ := $\emptyset$ \\
\verb|4. | \textbf{while} $len(waiting)>1$: \\
\verb|5. | \verb|  | \textbf{pick} some unexplored pair $(p_i,p_k)\in P \times waiting$ \\
\verb|6. | \verb|  | \textbf{compute} estimation $\widetilde{U}(p_i,p_k)$ \\
\verb|7. | \verb|  | \textbf{if} $\widetilde{U}(p_k,p_k)< d\cdot \widetilde{U}(p_i,p_k)$: \\
\verb|8. | \verb|  |  \verb|  | \textbf{remove} $p_k$ from $waiting$ \\
\verb|9. |  \verb|  | \textbf{else if} $\forall p_i'$ $\widetilde{U}(p_i', p_k)$ is already computed \\
\verb|10.| \verb|  | \verb|  | \textbf{remove} $p_k$ from $waiting$ \\
\verb|11.| \verb|  | \verb|  | \textbf{add} $p_k$ to $candidates$ \\
\verb|12.| \textbf{return} $\operatorname{argmax}_{p \in candidates}\min_{p' \in P}(\widetilde{U}(p, p)/\widetilde{U}(p', p))$
\end{small}
\end{algorithm}

As a first step, the algorithm computes estimations $\widetilde{U}(p',
p)$ for various $p'$ and $p$ and search for a parameter $p$ that maximizes the
value of $\min_{p' \in P}(\widetilde{U}(p, p)/\widetilde{U}(p', p))$.

Additionally, we speedup the search by introducing a heuristic
threshold $d$ (that is a parameter of our algorithm) and pruning
parameters $p$ such that $\widetilde{U}(p, p)/\widetilde{U}(p', p) <
d$.

Our algorithm~(see Algorithm \ref{alg:search}) starts with the
computation of estimations $\widetilde{U}(p_i, p_k)$ at diagonal
points (i.e. when $i=k$).  After that we iteratively pick a random
pair of strategies $(p_i, p_k)$ and compute $\widetilde{U}(p_i, p_k)$.
If $\widetilde{U}(p_k, p_k)/\widetilde{U}(p_i, p_k) < d$, then we
remove strategy $p_k$ from the further consideration and will never
consider again pairs of the form $(?, p_k)$.

We iterate the while-loop until we split all the parameters into those
$p$ for which we already computed the value of $\min_{p' \in
  P}(\widetilde{U}(p, p)/\widetilde{U}(p', p))$, and those, for which
we know that $\widetilde{U}(p, p)/\widetilde{U}(p', p) < d$ for some
$p'$. Then at line 12 we choose a strategy $p$ that maximizes $\min_{p' \in
  P}(\widetilde{U}(p, p)/\widetilde{U}(p', p))$

It should be noted, that if a threshold $d$ is large, then our algorithm can possibly return no result (because all the candidates will be filtered out). 
In this case one can retry with a smaller $d$ and reuse the already computed estimations.
If $d$ is equal to zero, then the algorithm is guaranteed to return a result.

\subsection{Evaluation of a Relevance of the Candidate}\label{evaluation_section}

Consider that after the first phase we selected a strategy $p$. Let
$H_{p, \delta}$ be a statistical hypothesis that $p$ is a
$\delta$-relaxed NE.  Now we want to find the maximal $\delta$ such
that we can accept the hypothesis $H_{p, \delta}$ with a given significance level
$\alpha$ (that is a parameter of the algorithm)\footnote{The significance level is a statistical parameter that defines the probability of accepting a hypothesis although it is actually false.}.

To do that we firstly reestimate $\widetilde{U}(p',p)$ for every
$p' \in P$ (possibly using the number of simulations that is different from the one that was used in the first phase, this will be discussed below).

Then we apply the following theorem:

\begin{theorem}\label{hyp_theorem}
  Suppose that for all $p' \in S$ we estimated $\widetilde{U}(p',  p)$ using $n$ random simulations, and $f(\delta) \leq \alpha$, where 
\begin{equation}\label{th:eq} f(\delta) \equiv \sum_{i=1..N} \frac{1}{2}\Big(1-\operatorname{erf}\big(\sqrt{n}({\widetilde{U}(p, p) - \delta \cdot \widetilde{U}(p_i, p)})\big)\Big)
\end{equation} and $\operatorname{erf}(x)=\frac{2}{\sqrt{\pi}}\int_{0}^x e^{-t^2} dt$ is Gauss error function.
    Then we can accept the statistical hypothesis $H_{p, \delta}$ with the significance level of $\alpha$.
\end{theorem}
\begin{proof}
Let ${q}_i = {U}(p_i, p_k)$ and $\widetilde{q}_i = \widetilde{U}(p_i, p_k)$.
Then $p_k$ satisfies $\delta$-relaxed Nash Equilibrium iff $\forall i \cdot q_k \geq \delta q_i$ 
Consider that we have $n$ Bernoulli random variables $\zeta_1, \zeta_2, \dots, \zeta_n$ where $Pr[\zeta_i=1]=q_i$. 
Consider that random variable $\xi_i$ is a mean of $n$ independent observations of $\zeta_i$, i.e.  $\xi_i = \sum_{j = 1..n}\zeta_i/n$
Then each $\widetilde{q}_i$ is an independent observation of $\xi_i$ and for large $n$ we have $\xi_i \sim \mathcal{N}(q_i, q_i(1-q_i)/n)$.
Probability of making type II error (\emph{accept} $H_{p_k, \delta}$
when it is false) is less or equal to
\begin{equation}\label{proof_eq}
 Pr[\xi_1=\widetilde{q}_1, \xi_2=\widetilde{q}_2, \dots, \xi_n=\widetilde{q}_N \| \bigvee_{i=1..N} q_k < \delta \cdot q_i]
\end{equation}
, that in turn is less or equal to
\begin{equation}
   \sum_{i=1..N} Pr[\xi_1=\widetilde{q}_1, \xi_2=\widetilde{q}_2, \dots, \xi_n=\widetilde{q}_N \| q_k < \delta \cdot q_i]
\end{equation}
, that in turn is less or equal to
\begin{equation}
   \sum_{i=1..N} Pr[\xi_i=\widetilde{q}_i, \xi_k=\widetilde{p}_k \| q_k < \delta \cdot q_i]
\end{equation}
, that in turn is less or equal to
\begin{equation}\label{eq4}
   \sum_{i=1..N} Pr[(\xi_k - \delta \cdot \xi_i) = (\widetilde{q}_k - \delta \cdot \widetilde{q}_i) \| q_k - \delta \cdot q_i< 0]
\end{equation}

For each $i$ we have:
\begin{equation}
(\xi_k - \delta \cdot \xi_i) \sim \mathcal{N}(q_k - \delta q_i, (q_k(1-q_k) + q_i(1-q_i))/n) 
\end{equation}

The truth of the theorem  follows from the fact that $q_k(1-q_k) + q_i(1-q_i) \leq 0.5$. \qed
\end{proof}

We apply this theorem in the following way. We first search for an integer-valued $b$ such that $f(b)<0$.
Then we use bisection numerical method to find a root of an equation $f(\delta) = \alpha$ on the interval $[0, b]$.
It can be easily seen that the function $f$ decreases and $f(0)>0$ and it implies that this $\delta$ satisfies the condition of the theorem.

\begin{figure*}[t]
\begin{center}
\begin{tabular}{cl}
\includegraphics[scale=1]{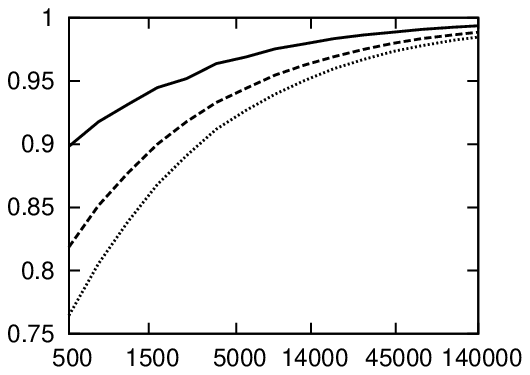}    & \includegraphics[scale=1]{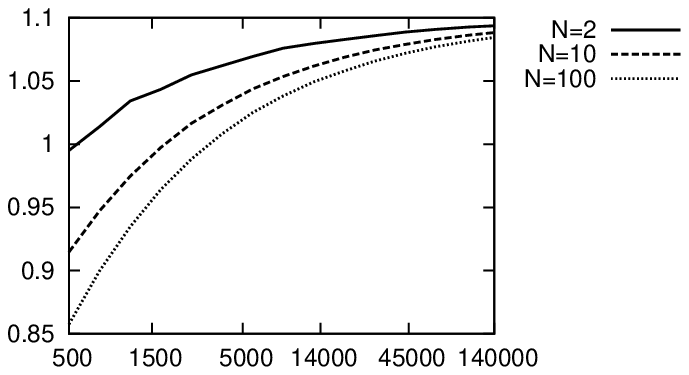} \\
\includegraphics[scale=1]{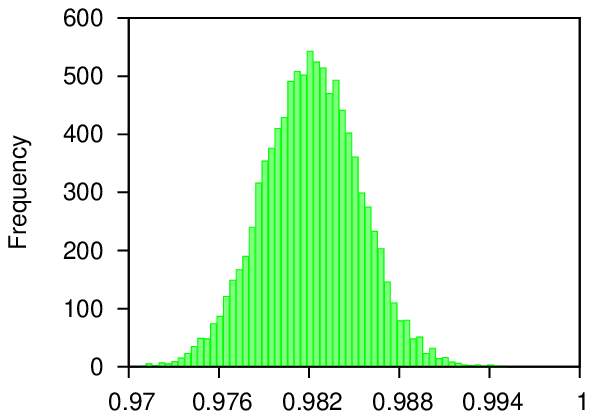}    & \includegraphics[scale=1]{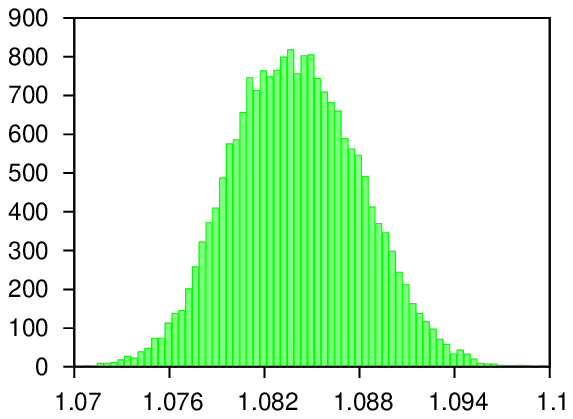}
\end{tabular}\caption{The experimental evaluation of the application of theorem \ref{hyp_theorem}. In the upper plots the $x$ axis denotes the number of simulations $n$, and $y$ axis denotes an average estimated value of $\delta$ computed using theorem \ref{hyp_theorem} for these numbers of simulations and the significance level of $\alpha=0.05$. $N$ is the total number of strategies. The bottom plots demonstrate the frequency distribution for the computed values of $\delta$ for $100000$ simulations and $100$ possible strategies. The left and right plots correspond to the cases when the real value of $\delta$ is equal to $1.0$ and $1.1$. } 
\label{fig:thmplots}
\end{center}
\end{figure*}

Our method provides only a lower bound for $\delta$ and the theorem \ref{hyp_theorem} does not state how many simulations are needed to compute a good estimation of $\delta$.
Indeed, we can compute statistically valid $\delta$ for any number of simulations.
And if the estimated value of $\delta$ is small, it can be a result of the fact that the number of simulations is insufficient or the real value of $\delta$ is small (or both).

Thus we performed an experiment to see how accurate is our method depending on the number of simulations.
We developed a simple model of our method, and in this model we assume that there are $N$ strategies $\{p_1, p_2, \dots p_N\}$, and we want to check that $p_1$ is a Nash Equilibrium.
The value of the utility function $U(p_1, p)$ is equal to $0.5$ for $p \neq p_1$, and it is equal to $0.5 \cdot \delta$ when $p = p_1$ (thus $p_1$ is a $\delta$-relaxed NE).
The experimental data for the cases when $\delta = 1.0$ and $\delta = 1.1$ is presented at Fig.~\ref{fig:thmplots}.
A reader can see that the results for these two cases are similar, and our other experiments (with $\delta$ ranging from $0.5$ to $2.0$) also demonstrate this similarity.
The accuracy of our method increases as the number of simulation $n$ increases or the number of possible strategies $N$ decreases. 
One can also see, that for this experiment $100000$ simulations seem to be enough to compute a good approximation for $\delta$ that is both statistically valid (this is ensured by the theorem~\ref{hyp_theorem}) and close to the real value of $\delta$.

\subsection{Implementation Details}

  We developed a tool written in Python programming language that
  implements the proposed
  algorithm.
  This tool uses \uppaal model checker as a simulation and monitoring engine for PWCTL properties.

  Our algorithm is based on Monte Carlo simulations and thus it is embarrassingly parallisable. 
  In our implementation we exploit this parallelisability by computing the estimations for different pairs of strategies on different nodes.
  The tool can be run on a cluster. One node acts as a master and picks the parameters $p_i$, $p_j$ that the slave nodes use to compute the estimations
  $\widetilde{U}(p_i, p_j)$.  The master node does not use any external
  job scheduler and submits jobs on its own using SSH connection to
  the computational nodes.  Currently we rely on the fact that the nodes
  share the same distributed file system, but in principle the master node
  can deploy all executables and models by itself.

\section{Results of Application}

In this paper we report on results of application of our tool to two
contention resolution protocols.  The first one is Aloha CSMA/CD
protocol that we model on a very abstract level and we'll describe our
model in details.  The second one is IEEE 802.15.4 CSMA/CA that we
model with a high precision and we'll just briefly sketch its
structure.

For both case studies we used the following parameters.
The number of simulations for estimation of utility function's values is equal to $10000$ for the first phase of our algorithm (searching for the best candidate for NE) and $100000$ for the second phase (evaluation of this candidate).
The value of $d$ parameter is equal to $0.9$, and the value of significance level $\alpha$ is equal to $0.05$.

All the experiments were performed on the $8$ node cluster, where each node has an Intel\textsuperscript{\copyright} Core\textsuperscript{\textregistered}2 Quad CPU.

\subsection{Application to Aloha CSMA/CD protocol}

\begin{figure}
\begin{center}
\includegraphics[scale=0.4]{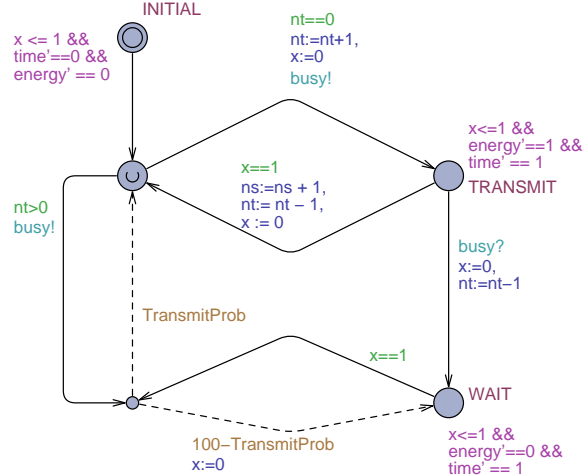}
\end{center}
\caption{Model of Aloha in \uppaal}
\label{fig:aloha}
\end{figure}

  Aloha protocol~\cite{aloha} is a simple Carrier Sense Multiple
  Access with Collision Detection (CSMA/CD) protocol that was used in
  the first known wireless data network developed at the University of
  Hawaii in 1971.
  
  In this protocol it is assumed that there are several nodes that
  share the same wireless medium.  Each node is listening to its own
  signal during its transmission and checks that the signal is not
  corrupted by another node's transmission.  In case of collision both
  nodes will stop transmitting immediately and wait for a random time
  before they'll try to transmit again.

  In our paper we consider unslotted Aloha in which the nodes are not
  necessary synchronized.  Additionally, we study k-persistent variant
  of Aloha, i.e. a protocol implementation in which a random delay
  before retransmission is distributed according to a geometric
  distribution.  This means that for each next time slot a node will
  transmit with probability \verb|TransmitProb| and will wait for one more
  slot (and then decide again) with probability $1-$\verb|TransmitProb|.  We
  assume that a node can change the value of \verb|TransmitProb|, thus a
  strategy of a node consists of choosing a value of \verb|TransmitProb|.
  We also assume that a node can use one out of $100$ of discretized values $\{0.01,
  0.02, \dots, 1\}$ of \verb|TransmitProb|.

  The \uppaal model of a single node is presented at Fig.~\ref{fig:aloha}.
  Wireless media is modeled using a broadcast channel \verb|busy|~(in
  which a signal is sent each time a new transmission starts) and
  integer variable \verb|nt| (that stores the number of stations that
  are currently transmitting).  Variable \verb|ns| stores the number of
  successful transmissions.  Time can pass only in locations
  \verb|INITIAL|, \verb|TRANSMIT| and \verb|WAIT|, two other locations
  are \emph{urgent}.  A node uses clocks \verb|x|, \verb|time| (that
  stores a time passed since the beginning) and \verb|energy| (that
  stores the amount of energy consumed, i.e. the amount of time spent
  in the location \verb|TRANSMIT|).
  
  We assume that there is a random uniformly distributed offset
  between the initial states of the nodes~(it is modeled by delay in
  location \verb|INITIAL|).  This may correspond to the situation,
  when there is a wireless sensor network and all sensors are aimed
  towards the same event.  As soon as this event happens, all the node
  will start transmission, but they will not be necessarily
  synchronized.

\begin{figure*}[t]
\begin{center}
\begin{tabular}{cc}
\includegraphics[scale=0.7]{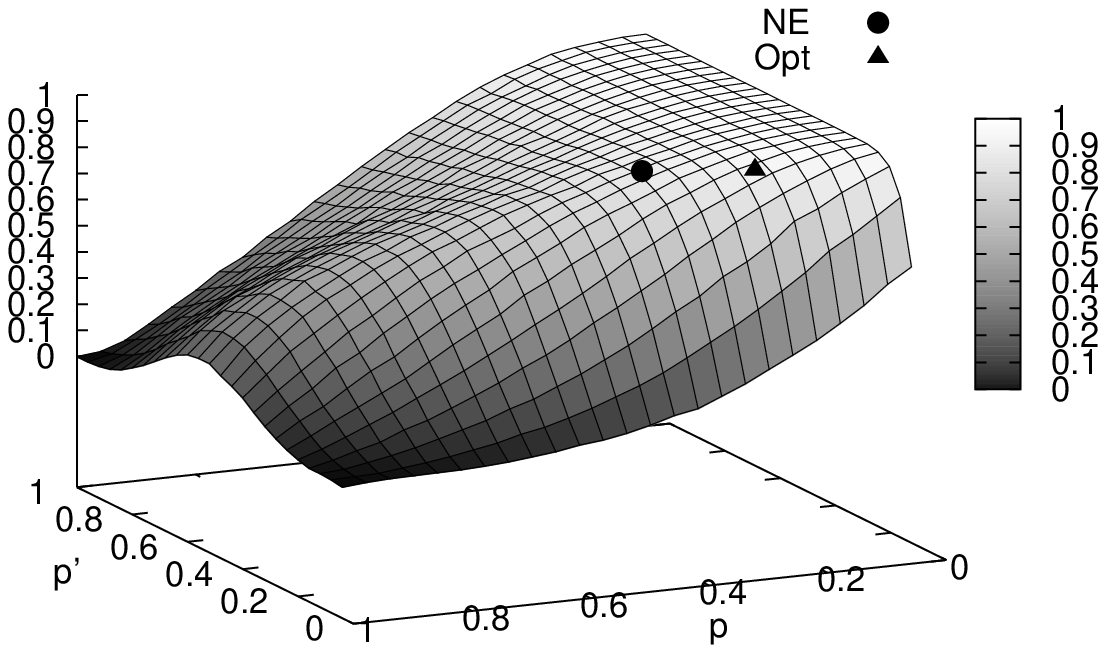}    & \includegraphics[scale=0.6]{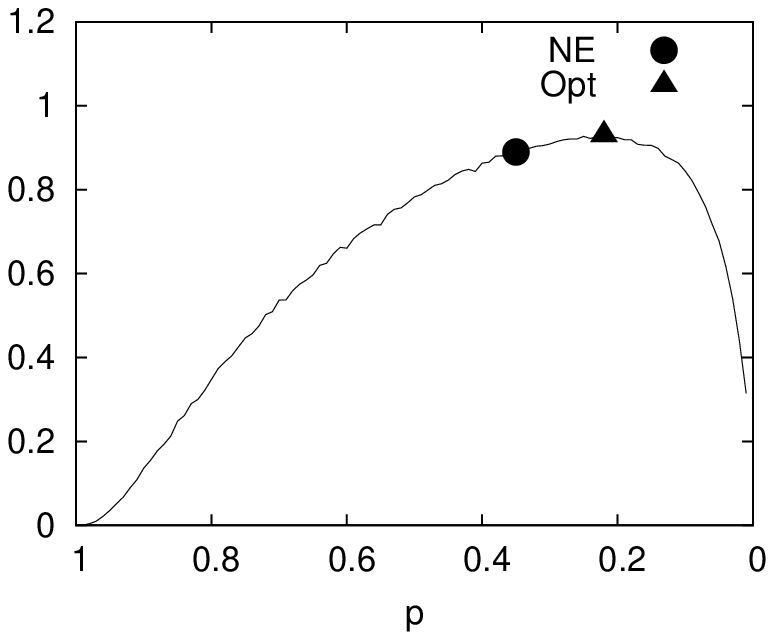} 
\end{tabular}
\caption{$\widetilde{U}(p', p)$ (left) and its diagonal slice (right) for Aloha with $5$ nodes}
\label{fig:alohaplots}
\end{center}
\end{figure*}

\begin{table*}
\begin{center}
 \begin{tabular}{@{\extracolsep{+5pt}}lcccccccc}
\toprule
    Number of nodes\                 		& 2    	& 3    	& 4    	& 5    	&  6    & 7   	& 8    	\\
\hline 
    $\delta$-relaxed NE strategy $p_{NE}$\    		& 0.37 	& 0.40 	& 0.35 	& 0.35 	& 0.41	& 0.42	& 0.41 	\\
    Value of $\delta$\               		& 0.992	& 0.980 & 0.992	& 0.990	& 0.993	& 0.992	& 0.987	\\
    $\widetilde{U}(p_{NE}, p_{NE})$ 				\ 	& 0.99  & 0.98  & 0.95  & 0.89	& 0.75	& 0.61  & 0.50  \\
    Symmetric optimal strategy $p_{opt}$\               			& 0.30 	& 0.30 	& 0.26 	& 0.22	& 0.19	& 0.15 	& 0.14 	\\
    $\widetilde{U}(p_{Opt}, p_{opt})$\ 	& 0.99 	& 0.98  & 0.96 	& 0.90	& 0.87  & 0.98  & 0.76 		\\
    Computation time \ 					& 2m5s 	& 3m44s & 7m62s & 15m45s& 26m11s& 37m55s& 59m15s\\
\bottomrule
 \end{tabular}
\end{center}
\caption{Nash equilibrium (NE) and Symmetric optimal (Opt) strategies for Aloha}
\label{aloha:results}
\end{table*}	

In our experiments we assumed that the goal of a node is to transmit a
single frame within $50$ time units and to limit energy consumption by
$3$.  This goal can be expressed using the following PWCTL formula:

\begin{equation}
  \Diamond_{\text{{\tt Node(0).time}} \leq \text{{\tt 50}}}(\text{{\tt Node(0).ns}} \geq \text{{\tt 1}} \wedge \text{{\tt Node(0).energy}} \leq \text{{\tt 3}})
\end{equation}

It should be noted, that even our (unslotted) Aloha model looks simple, we can't propose an analytical way of computing $U(p', p)$ for a given values of $p'$ and $p$.
The problem is that our model works in real-time and we can't decompose its behavior into rounds (slots) and compute $U(p', p)$ recursively based on the nodes' actions in the current round and values of $U(p', p)$ in the next possible rounds~(like it was done in~\cite{Mackenzie2001} for \emph{slotted} Aloha).


Fig.~\ref{fig:alohaplots} depicts the plot of the utility function estimation $\widetilde{U}(p', p)$ for the first player for the network of $5$ nodes (remind, that $p'$ is a strategy of the first player, and $p$ is common strategy of all the other players). 
It also shows Nash Equilibrium  (NE) and symmetric optimal (Opt) strategies.
It should be noted, that due to the usage of $d$ parameter our algorithm didn't compute $\widetilde{U}(p', p)$ for all possible $p$ and $p'$ (in fact, only $3742$ out of $10000$ values were computed).

Intuitively, a Nash Equilibrium for Aloha exists, because a node has to satisfy both time and energy constraints.
When the honest nodes use the value of \verb|TransmitProb| that is close to $1$, it forces the selfish node to use a smaller value of \verb|TransmitProb| to bound the number of collisions (and hence the energy consumption). When the default value of \verb|TransmitProb| is close to $0$, the selfish node uses a larger value of \verb|TransmitProb| to decrease the expected time before the next retransmission, since the probability of a collision is small for this case.
This ensures that a Nash Equilibrium strategy exists in between $0$ and $1$.

Table~\ref{aloha:results} contains the results for ALOHA with different number of nodes.
It can be seen, that relaxed NE and symmetric optimal strategies coincide for the case of two network nodes, but for the networks with more nodes relaxed NE is less efficient than symmetrical optimal strategy.

\subsection{Application to IEEE 802.15.4 CSMA/CA Protocol}

IEEE 802.15.4 standard~\cite{IEEE.802.15.4} specifies the physical
layer and media access control layer for low-cost and low-rate
wireless personal area networks.  Upper layers are not covered by IEEE
802.15.4 and are left to be extended in industry and individual
applications.  One of such extensions is ZigBee~\cite{zigbee} that
together with IEEE 802.15.4 completes description of a network stack.
Typical applications of ZigBee include smart home control and wireless
sensor networks.


We applied our tool to the analysis of Multiple Access/Collision
Avoidance (CSMA/CA) network contention protocol being a part of IEEE
802.15.4.  Unlike Aloha, the IEEE 802.15.4 standard assumes that a
wireless node can't listen to its own transmission and thus it is not
possible to detect a collision as soon as it occurs and stop
transmission.  A node will detect a collision later when it does not
receive an acknowledgment within a given time bound.  Before each
transmission a node performs a Clear Channel Assessment (CCA),
i.e. checks that no other node is transmitting.  If CCA was not
successful (the medium was busy), then the node waits for a random time
before performing CCA again, and this time is distributed according to
the binary exponential backoff mechanism (that is controlled by the
parameters \verb|MinBE|, \verb|MaxBE| and \verb|UnitBackoff| in our
model).  If CCA was successful (the medium was clear), then the node
switches to the transmitting mode and starts transmission.  However,
this switching takes non-zero time (\verb|TurnAround| in our model),
and another node can start transmitting during this period, that will
lead to a collision.

The standard defines both slotted (with beacon
synchronization) and unslotted modes of CSMA/CA; in our paper we
consider only unslotted one.

\begin{figure*}[t]
\includegraphics[width=0.98\textwidth]{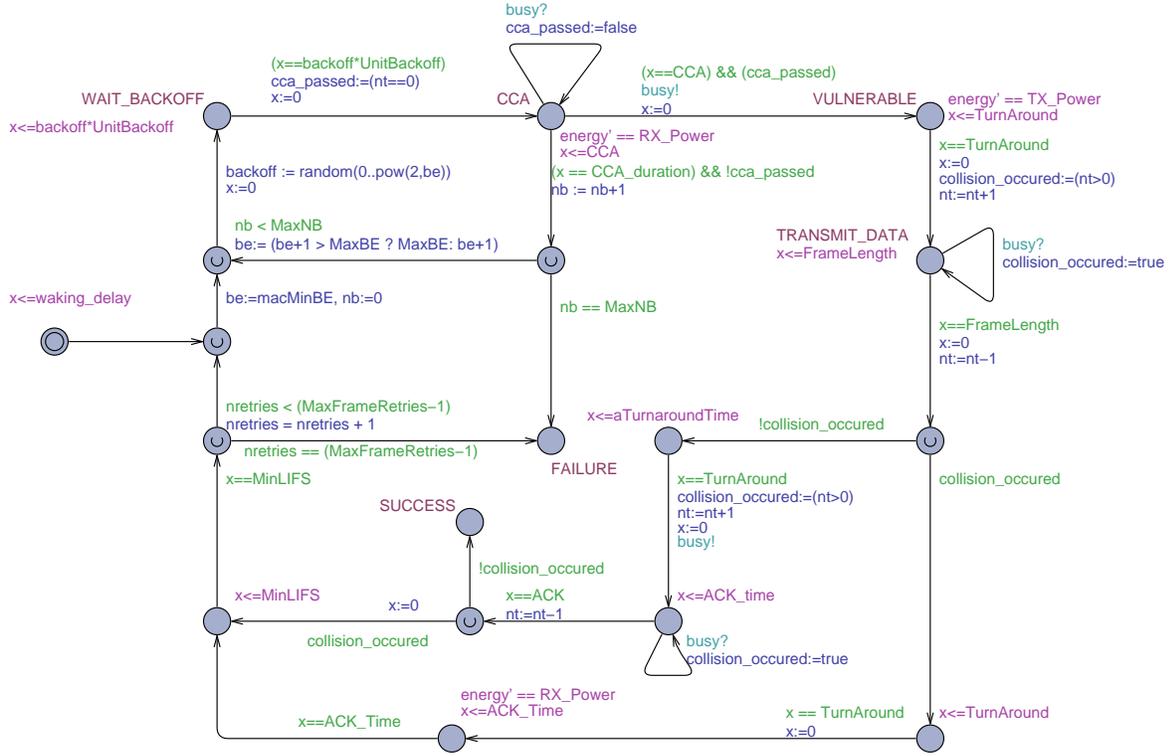}
\caption{Model of IEEE 802.15.4 CSMA/CA}
\label{fig:model}
\end{figure*}

The model of a single node operating according to IEEE 802.15.4
CSMA/CA is depicted at Fig.~\ref{fig:model}.  The values of
\verb|MinBE|, \verb|MaxBE|, \verb|MaxFrameRetries|, \verb|TurnAround|
were taken from the IEEE 802.15.4 standard assuming that the network
is operating on baud rate 20kbps and on 868 Mhz
band. \verb|FrameLength| is considered to be 35 bytes (including 25
bytes for ZigBee header and 10 bytes for the valuable information).
We assume that the frame size is 35 bytes (25 bytes for ZigBee header
and 10 bytes for the actual data).  Energy consumption constraints
\verb|TX_Power| and \verb|RX_Power| were taken from the specification
of U-Power 500 chip (54 mA and 26 mA operating on 3.0V respectively).

We assume that a node can change the value of \verb|UnitBackoff|
parameter. This parameter linearly scales the binary exponential backoff scheme.  
If its value is equal to $0$, then a node will try to
transmit as soon as it wants to.  The large values of
\verb|UnitBackoff| corresponds to large delays before transmission.
We consider that the possible values of \verb|UnitBackoff| are $\{0, 1, 2, \dots, 50\}$.
We assume that the goal of a node for CSMA/CA is similar to the one used in the Aloha case study (i.e. to transmit a frame within the given time and energy bounds).

\begin{figure*}[t]
\begin{center}
\begin{tabular}{cc}
\includegraphics[scale=0.6]{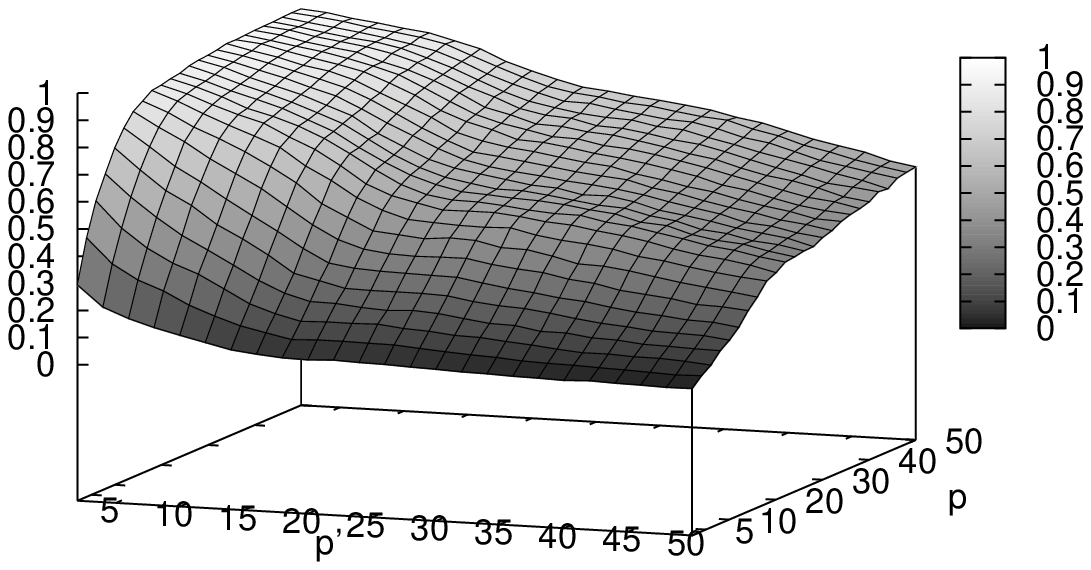} & \includegraphics[scale=0.6]{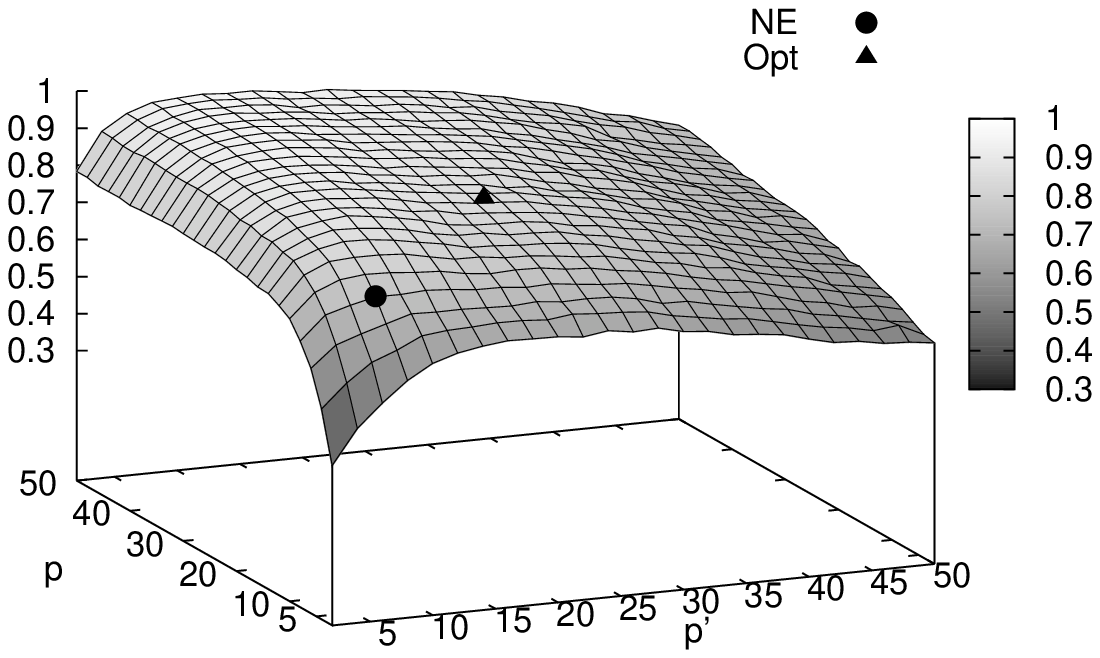}    
\end{tabular}
\caption{$\widetilde{U}(p', p)$ for CSMA/CA for $5$ nodes without(left) and with(right) coalitions}
\label{fig:csmacaplots}
\end{center}
\end{figure*}

Our tool detected a trivial NE \verb|UnitBackoff|=0, see the plot at Fig. \ref{fig:csmacaplots} (left) for an illustration.
It means that a selfish node will always try to transmit as soon as possible by choosing \verb|UnitBackoff|=0.
This coincides with the results of~\cite{Cagalj04oncheating} obtained for IEEE 802.11 CSMA/CA protocol.
Intuitively, it is always profitable to transmit as soon as possible since if a selfish node will retransmit just after the collision, the rest (honest) nodes will probably
detect this during the Clear Channel Assessment procedure and they will not corrupt the retransmission of the selfish node.

\begin{table*}
\begin{center}
 \begin{tabular}{@{\extracolsep{+5pt}}lccccccc}
\toprule
    Number of nodes in one coalition\                 		& 1    	& 2    	& 3    	& 4    	&  5    \\ 
\hline 
    $\delta$-relaxed NE strategy $p_{NE}$\			& 11 	& 8 	& 15 	& 25 	& 28	\\ 
    Value of $\delta$\               				& 0.900	& 0.985 & 0.986	& 0.990	& 0.990	\\ 
    $\widetilde{U}(p_{NE}, p_{NE})$	 \ 			& 0.86  & 0.76  & 0.81  & 0.85	& 0.83	\\ 
    Symmetric optimal strategy $p_{opt}$\       		& 13 	& 23 	& 31 	& 34	& 48	\\ 
    $\widetilde{U}(p_{opt}, p_{opt})$	 \ 			& 0.87	& 0.85 	& 0.87 	& 0.87	& 0.86  \\ 
    Time \ 							& 1m08s	& 5m45s & 7m62s & 32m49s& 57m59s \\ 
\bottomrule
 \end{tabular}
\end{center}
\caption{Nash equilibrium (NE) and Symmetric optimal (Opt) strategies for CSMA/CA with coalitions}
\label{csma:results}
\end{table*}

In order to illustrate our algorithm we also considered the situation when network nodes (game players) form coalitions.
It can correspond to the situation when several network devices belong to the same user and it will not be profitable for the user if these devices compete with each other.
The intuition is that players of the same coalition will not choose ``always transmit'' strategy because in this case they will disturb each other.
This is confirmed by plot at Fig. \ref{fig:csmacaplots} (right) and table \ref{csma:results}, where we considered the case of two coalitions of the same size.

\section{Related Work}

The paper \cite{Mackenzie2001} is the first one that applies the concept of Nash Equilibrium to the analysis of Medium Access and power control games in \emph{slotted} Aloha protocol.
Later this approach has been applied to the most of the layers of a network stack: to the Physical~\cite{EURECOM_1656,Mackenzie2001,cognitive}, Medium Access~\cite{Rao2008,Cagalj04oncheating,TransmissionProbability,conf/pimrc/GhasemiF08}, Network~\cite{Felegyhezi06nashequilibria,Zakiuddin200567} and Application~\cite{LCA-BOOK-2006-001} layers.

Although our approach can be in principle applied to any network
layer, it is particularly well suited for the random access Medium
Access layer protocols, since such protocols possess probabilistic
behavior (here we can use our Weighted Timed Automata semantics) and
work in real-time.  In this settings, our SMC-based approach extends
the manual analytical approach, that can be complicated, error-prone
and typically applied to slotted (discrete time) protocols
only~\cite{TransmissionProbability,Mackenzie2001}.  On the other hand,
our approach extends the simulation-based approach (for instance,
\cite{Cagalj04oncheating}), since we formally describe a modeling
formalism for which we can provide a confidence on the results.

Additionally, in our paper we use the expressive PWCTL logic to
express the goals of the network nodes, and thus to define their
utility functions with respect to time and energy constraints.  This
allows us to apply the same framework to the analysis of different
protocols, while another approaches does not allow such a
generalization.

Our experimental results extend those proposed in \cite{Mackenzie2001}
from the \emph{slotted} Aloha to the unslotted one.  Up to our
knowledge, we are also the first ones, who study coalitions between
nodes in the IEEE 802.15.4 CSMA/CA protocol.

\section{Conclusions}

In this paper we have presented a methodology to apply statistical
model checking to search for a Nash equilibrium on different types of
networks.  Experiments demonstrate the maturity of our technique and
shows that it can be applied in principle to more complex problems.
The technique avoids analytical analysis of the model and contrary to
pure simulation-based techniques, ours provides statistical confidence
on its results.  As future work we will extend the language of our
tool to be able to apply it to other domains such as biological
systems.

\bibliographystyle{eptcs}
\bibliography{main}

 \end{document}